\documentclass[letterpaper, 10pt, conference]{ieeeconf}

\IEEEoverridecommandlockouts
\usepackage{textcomp, amsmath, amssymb, amsfonts, algorithm}
\usepackage{algpseudocode, graphicx, siunitx, bm, float, booktabs}
\usepackage[colorlinks=true, allcolors=black]{hyperref}

\newtheorem{theorem}{Theorem}
\newtheorem{proposition}{Proposition}
\newtheorem{definition}{Definition}
\newtheorem{assumption}{Assumption}
\newtheorem{remark}{Remark}
\newtheorem{example}{Example}

\title{\LARGE \bf CBF-based Probabilistic Safe Navigation under Unknown Nonlinear Obstacle Dynamics}

\author{Jiwon Lee, Hugo Matias, Daniel Silvestre, Thinh T. Doan
\thanks{This work was supported in part by NSF under CAREER Award 2527059, UT Austin through the UT-Portugal 2024 Extra Exploratory Project, and CTS and LASI under grant CTS UID/0066/2025.}
\thanks{Jiwon Lee and Thinh T. Doan are with the Department
of Aerospace Engineering and Engineering Mechanics, University of Texas at Austin, Austin, USA (e-mails: jiwon.lee@utexas.edu and thinhdoan@utexas.edu)}
\thanks{H. Matias and D. Silvestre are with the Center of Technology and Systems (UNINOVA-CTS), 
NOVA School of Science and Technology (NOVA-FCT), 2829-516 Caparica, Portugal, 
and also with the Institute for Systems and Robotics (ISR-Lisbon), Instituto 
Superior Técnico (IST), 1049-001 Lisbon, Portugal 
(emails: h.matias@campus.fct.unl.pt and dsilvestre@fct.unl.pt).}
}

\begin{document}

\maketitle
\thispagestyle{empty}
\pagestyle{empty}


\begin{abstract}
Safe navigation for an ego vehicle in uncertain environments characterized by dynamic obstacles with unknown nonlinear dynamics is a challenging problem of significant practical interest. Existing approaches in the literature either lack formal safety guarantees, require full model knowledge, or fail to account for the risk associated with the vehicle's exact body geometry and the temporal evolution of uncertainty between sampling instants. In this paper, we propose a data-driven observer for the unknown obstacle dynamics that generates an $\alpha$-confidence set flow, which is exactly transformed into a Control Barrier Function (CBF) to enforce ($1-\alpha$)-probability safety. The proposed framework accommodates nonlinear ego vehicle dynamics of arbitrary relative degree, as demonstrated through case studies involving first- and second-order dynamics of an unmanned surface vehicle.
\end{abstract}



\section{INTRODUCTION}
Safe obstacle avoidance is a fundamental requirement for autonomous systems operating in the physical world. In this context, an autonomous agent -- whether a ground vehicle, aerial robot, or marine platform -- must navigate efficiently toward its objective while continuously ensuring collision-free operation in the presence of surrounding obstacles. This task becomes particularly challenging in dynamic and uncertain environments, where obstacles may move unpredictably and their states are not directly accessible. In practice, the controller often has access only to noisy, intermittent, or partial measurements of obstacle positions and velocities, making it impossible to guarantee safety using nominal trajectories or deterministic predictions alone. Reliable navigation therefore demands a control framework that explicitly accounts for measurement uncertainty, unknown obstacle dynamics, and the temporal evolution of moving obstacles between sampling instants.

A large body of prior work has addressed dynamic obstacle avoidance through reactive geometric methods such as velocity obstacles \cite{fiorini1998motion} and the dynamic window approach \cite{fox2002dynamic}, which remain influential because of their computational simplicity and online applicability.
Another line of work represents the surrounding environment probabilistically using occupancy grids and their dynamic extensions, enabling uncertainty-aware mapping and planning in changing environments \cite{elfes1990occupancy, meyer2012occupancy}.
More recently, stochastic model predictive control formulations have incorporated obstacle uncertainty directly into collision-avoidance constraints \cite{andersson2016model}.
While these approaches have demonstrated practical success, they typically face one of two limitations: either the obstacle uncertainty representation is not naturally integrated into a formal safety certificate, or the resulting safety treatment requires approximations that introduce conservatism.

Control barrier functions (CBFs) provide a rigorous framework for enforcing safety through forward invariance of a safe set and have become a standard tool for safety-critical control \cite{ames2016control}.
Recent work has extended this framework to uncertain and dynamic environments. For example, observer-based robust CBF methods have been proposed for moving-obstacle avoidance by incorporating obstacle-state observers and deterministic bounds on measurement and prediction uncertainty \cite{quan2025observer}.
Stochastic CBF formulations have also been developed to provide safety guarantees for systems with stochastic dynamics \cite{nishimura2024control}.
These advances significantly broaden the applicability of CBF-based safety filtering. 
However, they generally presume that the uncertainty entering the barrier constraint is already available in a form suitable for control design, such as a known deterministic bound, or a prescribed stochastic model.
In this case, the main challenge is to construct, from data, a time-varying obstacle occupancy set that is probabilistically meaningful and can be incorporated into a real-time barrier-based safety filter.

Set-based estimation offers a natural language for representing such obstacle uncertainty.
In this context, constrained convex generators (CCGs) have recently emerged as a flexible convex-set representation with useful closure properties under affine maps and Minkowski sums, making them attractive for estimation and navigation problems \cite{silvestre2021constrained}.
Most directly related to the present work, a sampled-data safe-navigation framework is developed in which uncertain linear obstacle dynamics are estimated and propagated as CCG-valued flows, followed by a systematic conversion to CBFs \cite{matias2026safe}.
That work shows the value of preserving set geometry throughout the safety construction rather than replacing it by simpler bounds. 
At the same time, its estimation framework assumes uncertain linear obstacle dynamics, whereas the present problem considers moving obstacles with unknown dynamics observed only through noisy position outputs.

\noindent\textbf{Main Contributions.} We develop a data-driven safety framework that connects online obstacle estimation with CBF-based collision avoidance. Unlike many existing CBF methods that rely on fixed uncertainty bounds or model-based obstacle observers, the proposed approach uses a sliding-window least-squares estimator to infer obstacle motion directly from recent noisy measurements and to construct a time-varying probabilistic obstacle set. This set is then represented in CCG form and incorporated into the barrier function design, enabling safety filtering that adapts online to measurement-informed obstacle uncertainty.



\subsubsection*{Notation} For two vectors $\mathbf{x}_1 \in\mathbb{R}^{n_1}$, 
$\mathbf{x}_2 \in \mathbb{R}^{n_2}$, we often use 
$(\mathbf{x}_1, \mathbf{x}_2) \in \mathbb{R}^{n_1+n_2}$ to denote their 
concatenation. For a differentiable function 
$h: \mathbb{R}^n\times[t_0, t_\mathrm{f}) \rightarrow \mathbb{R}$ and a mapping 
$\mathbf{G}: \mathbb{R}^n \rightarrow \mathbb{R}^{n \times m}$, we let 
$L_\mathbf{G}h(\mathbf{x}, t) = 
\frac{\partial}{\partial\mathbf{x}}h(\mathbf{x}, t) \mathbf{G}(\mathbf{x})$. 
For a set-valued function 
$\mathcal{D}: [t_0, t_\mathrm{f}) \rightrightarrows \mathbb{R}^n$, we denote 
its graph by 
$\mathcal{G}(\mathcal{D}) \subseteq \mathbb{R}^n\times[t_0, t_\mathrm{f})$. 
Finally, $\mathbf{0}_{n\times m}$ is the $n\times m$ matrix of zeros, and 
$\mathbf{I}_{n}$ is the identity matrix of size $n$.

\section{PRELIMINARIES} \label{Sec:Preliminaries}

This section provides essential background on CBFs and CCGs, which will play a 
central role throughout this letter.

\subsection{Control Barrier Functions}

Consider a nonlinear control-affine system of the form
\begin{equation}
    \Dot{\mathbf{x}} = 
    \mathbf{f}(\mathbf{x}) + \mathbf{G}(\mathbf{x})\mathbf{u},
    \label{Eq:CAS}
\end{equation}
where $\mathbf{x} \in \mathbb{R}^n$ is the state, $\mathbf{u} \in \mathbb{R}^m$ 
is the control input, and the functions 
$\mathbf{f}: \mathbb{R}^n \rightarrow \mathbb{R}^n$ and 
$\mathbf{G}: \mathbb{R}^n \rightarrow \mathbb{R}^{n \times m}$ are locally 
Lipschitz. Let 
$\mathbf{k}: \mathcal{G}(\mathcal{D}) \rightarrow \mathbb{R}^m$ be a feedback 
controller, defined over the graph of a set-valued flow 
$\mathcal{D}: [t_0, t_\mathrm{f}) \rightrightarrows \mathbb{R}^n$. Applying 
$\mathbf{k}$ to \eqref{Eq:CAS} yields the time-varying closed-loop system
\begin{equation}
    \Dot{\mathbf{x}} = 
    \mathbf{f}(\mathbf{x}) + \mathbf{G}(\mathbf{x})\mathbf{k}(\mathbf{x}, t).
    \label{Eq:CAS-CL}
\end{equation}
Since $\mathbf{f}$ and $\mathbf{G}$ are locally Lipschitz, if $\mathbf{k}$ is 
locally Lipschitz in $\mathbf{x}$ and continuous in $t$, then for any
initial state $\mathbf{x}_0 \in \mathcal{D}(t_0)$, there exists a unique 
solution $\bm{\varphi}: I(\mathbf{x}_0) \rightarrow \mathbb{R}^n$ satisfying 
\eqref{Eq:CAS-CL} on a maximal interval of existence 
$I(\mathbf{x}_0) \subseteq [t_0, t_\mathrm{f})$. From now on, we will assume 
that $I(\mathbf{x}_0) = [t_0, t_\mathrm{f})$ for convenience. Next, we revisit 
the notion of forward invariance.

\begin{definition}
    A set-valued flow 
    $\mathcal{C}: [t_0, t_\mathrm{f}) \rightrightarrows \mathbb{R}^n$ is said 
    to be forward invariant for \eqref{Eq:CAS-CL} if, for every initial 
    condition $\mathbf{x}_0 \in \mathcal{C}(t_0)$, we have 
    $\bm{\varphi}(t) \in \mathcal{C}(t)$ for all $t \in [t_0, t_\mathrm{f})$.
\end{definition}

We intend to ensure that, at every time instant, the solution of 
\eqref{Eq:CAS-CL} lies in a safe set $\mathcal{C}(t) \subset \mathbb{R}^n$, 
which amounts to ensuring forward invariance of a set-valued flow 
$\mathcal{C}: [t_0, t_\mathrm{f}) \rightrightarrows \mathbb{R}^n$ for 
\eqref{Eq:CAS-CL}. Particularly, we consider a set-valued flow $\mathcal{C}$ 
defined as
\begin{equation}
    \mathcal{C}(t) = \{\mathbf{x} \in \mathbb{R}^n: h(\mathbf{x}, t) \geq 0\},
    \label{Eq:SafeFlow}
\end{equation}
where $h: \mathbb{R}^n\times[t_0, t_\mathrm{f}) \rightarrow \mathbb{R}$ is 
continuously differentiable with 
$\frac{\partial}{\partial\mathbf{x}}h(\mathbf{x}, t) \neq \mathbf{0}$ when 
$h(\mathbf{x}, t) = 0$. If $h$ has the properties of a CBF, then it can be 
used to ensure safety for \eqref{Eq:CAS}.

\begin{definition}[CBF \cite{ames2016control}]
    Let $\mathcal{C}: [t_0, t_\mathrm{f}) \rightrightarrows \mathbb{R}^n$ be 
    defined by \eqref{Eq:SafeFlow}, for a continuously differentiable 
    $h: \mathbb{R}^n\times[t_0, t_\mathrm{f}) \rightarrow \mathbb{R}$ with
    $\frac{\partial}{\partial\mathbf{x}}h(\mathbf{x}, t) \neq \mathbf{0}$ when 
    $h(\mathbf{x}, t) = 0$. Function $h$ is a CBF for \eqref{Eq:CAS} if there 
    exists $\mathcal{D}: [t_0, t_\mathrm{f}) \rightrightarrows \mathbb{R}^n$ 
    with $\mathcal{C}(t) \subseteq \mathcal{D}(t)$ for all 
    $t \in [t_0, t_\mathrm{f})$ and an extended class-$\mathcal{K}_\infty$ 
    function $\alpha: \mathbb{R} \rightarrow \mathbb{R}$ such that, for all 
    $(\mathbf{x}, t) \in \mathcal{G}(\mathcal{D})$,
    \begin{equation}
        \sup_{\mathbf{u} \in \mathbb{R}^m}
        \Dot{h}(\mathbf{x}, t, \mathbf{u}) > - \alpha(h(\mathbf{x}, t)),
        \label{Eq:DefCBF}
    \end{equation}
    where $\Dot{h}(\mathbf{x}, t, \mathbf{u}) = 
    L_\mathbf{f}h(\mathbf{x}, t) + L_\mathbf{G}h(\mathbf{x}, t)\mathbf{u} 
    + \frac{\partial}{\partial t}h(\mathbf{x}, t)$. 
\end{definition}

\newpage

For a CBF $h$ for \eqref{Eq:CAS} and an associated extended 
class-$\mathcal{K}_\infty$ function $\alpha$, we define the pointwise set of 
controls 
\begin{equation}
    K_\mathrm{CBF}(\mathbf{x}, t) = \left\{\mathbf{u} \in \mathbb{R}^m: 
    \Dot{h}(\mathbf{x}, t, \mathbf{u}) \geq -\alpha(h(\mathbf{x}, t))\right\}.
\end{equation}
This yields the following main result concerning CBFs.

\begin{theorem}[Safe Control \cite{ames2016control}]
    Let $\mathcal{C}: [t_0, t_\mathrm{f}) \rightrightarrows \mathbb{R}^n$ be 
    defined as in \eqref{Eq:SafeFlow}. If $h$ is a CBF for \eqref{Eq:CAS} on 
    $\mathcal{D}: [t_0, t_\mathrm{f}) \rightrightarrows \mathbb{R}^n$, 
    
    \noindent then $K_\mathrm{CBF}(\mathbf{x}, t)$ is nonempty for all 
    $(\mathbf{x}, t) \in \mathcal{G}(\mathcal{D})$, and any controller 
    $\mathbf{k}: \mathcal{G}(\mathcal{D}) \rightarrow \mathbb{R}^m$ locally 
    Lipschitz in $\mathbf{x}$, continuous in $t$, with 
    $\mathbf{k}(\mathbf{x}, t) \in K_\mathrm{CBF}(\mathbf{x}, t)$ for all 
    $(\mathbf{x}, t) \in \mathcal{G}(\mathcal{D})$ renders $\mathcal{C}$ 
    forward invariant for the resulting closed-loop system.
    
\end{theorem}

A major utility of CBFs is their ability to serve as a safety filter for a 
nominal controller $\mathbf{k}_\mathrm{d}: 
\mathbb{R}^n\times [t_0, t_\mathrm{f}) \rightarrow \mathbb{R}^m$. Given a CBF 
$h: \mathbb{R}^n\times[t_0, t_\mathrm{f}) \rightarrow \mathbb{R}$ for 
\eqref{Eq:CAS} on 
$\mathcal{D}: [t_0, t_\mathrm{f}) \rightrightarrows \mathbb{R}^n$, the typical 
approach for designing a safety filter 
$\mathbf{k}: \mathcal{G}(\mathcal{D}) \rightarrow \mathbb{R}^m$ is 
through the following QP: 
\begin{equation} 
    \begin{aligned}
        \mathbf{k}(\mathbf{x}, t) = 
        \underset{\mathbf{u} \in \mathbb{R}^m}{\arg\min}\,\,
        &\frac{1}{2}\|\mathbf{u} - \mathbf{k}_\mathrm{d}(\mathbf{x}, t)\|^2\\
        \text{subject to}\,\, &\Dot{h}(\mathbf{x}, t, \mathbf{u}) \geq -
        \alpha(h(\mathbf{x}, t)),
    \end{aligned}
    \label{Eq:SafetyFilter}
\end{equation}
where $\alpha$ is an associated extended class-$\mathcal{K}_\infty$ function. 
Such a controller admits a closed-form expression, and it is locally Lipschitz 
in $\mathbf{x}$ and continuous in $t$ if the nominal controller, the CBF 
gradient, and $\alpha$ have the same properties \cite{cohen2023characterizing}, 
\cite{morris2013sufficient}. 

\subsection{Constrained Convex Generators}

In this letter, we will also rely on CCGs, which constitute the 
state-of-the-art framework for representing convex sets.
\begin{definition}[CCG \cite{silvestre2021constrained}]
    A set $\mathcal{Z} \subset \mathbb{R}^n$ is a CCG if there exists a tuple 
    $(\mathbf{G}, \mathbf{c}, \mathbf{A}, \mathbf{b}) \in 
    \mathbb{R}^{n\times\xi}\times\mathbb{R}^n\times\mathbb{R}^{c\times\xi}
    \times\mathbb{R}^{c}$ and a generator set $\mathfrak{G} = 
    \mathcal{G}_1\times\mathcal{G}_2\times\dots\times\mathcal{G}_G \subset 
    \mathbb{R}^\xi$ such that
    \begin{equation}
        \mathcal{Z} = \{\mathbf{G}\bm{\xi} + \mathbf{c}: 
        \mathbf{A}\bm{\xi} = \mathbf{b}, \bm{\xi} \in \mathfrak{G}\},
        \label{Eq:DefinitionCCG}
    \end{equation}
    where, for each $i \in \{1, \dots, G\}$, the set $\mathcal{G}_i$ is defined 
    as the zero-sublevel set of a convex function 
    $g_i: \mathbb{R}^{\xi_i} \rightarrow \mathbb{R}$.
\end{definition}

For brevity, we write $\mathcal{Z} = (\mathbf{G}, \mathbf{c}, \mathbf{A}, 
\mathbf{b}, \mathfrak{G}) \subset \mathbb{R}^n$. CCGs are appealing because 
they represent highly general convex sets while supporting key set operations 
such as affine maps, Minkowski sums, and generalized intersections.

\setlength{\arraycolsep}{1pt}
\begin{proposition}[Operations with CCGs \cite{silvestre2021constrained}]
    Consider a matrix $\mathbf{R} \in \mathbb{R}^{y\times n}$, a vector 
    $\mathbf{t} \in \mathbb{R}^y$, and three CCGs:
    \begin{itemize}
        \item $\mathcal{Z} = 
        (\mathbf{G}_z, \mathbf{c}_z, \mathbf{A}_z, \mathbf{b}_z, \mathfrak{G}_z)
        \subset \mathbb{R}^n$;
        \item $\mathcal{W} = 
        (\mathbf{G}_w, \mathbf{c}_w, \mathbf{A}_w, \mathbf{b}_w, \mathfrak{G}_w)
        \subset \mathbb{R}^n$;
        \item $\mathcal{V} = 
        (\mathbf{G}_v, \mathbf{c}_v, \mathbf{A}_v, \mathbf{b}_v, \mathfrak{G}_v)
        \subset \mathbb{R}^y$.
    \end{itemize}
    Then, the following identities hold:
    \begin{equation}
        \begin{aligned}
            \mathbf{R}\mathcal{Z} + \mathbf{t} &= (\mathbf{R}\mathbf{G}_z, 
            \mathbf{R}\mathbf{c}_z + \mathbf{t}, \mathbf{A}_z, \mathbf{b}_z, 
            \mathfrak{G}_z),\\[1mm]
            \mathcal{Z} \oplus \mathcal{W} &= 
            \scalebox{0.84}{$\left(
            \begin{bmatrix}
                \mathbf{G}_z & \mathbf{G}_w
            \end{bmatrix}\hspace{-1mm},
            \mathbf{c}_z + \mathbf{c}_w, 
            \begin{bmatrix}
                \mathbf{A}_z & \mathbf{0}\\
                \mathbf{0} & \mathbf{A}_w
            \end{bmatrix}\hspace{-1mm},
            \begin{bmatrix}
                \mathbf{b}_z\\
                \mathbf{b}_w
            \end{bmatrix}\hspace{-1mm},
            \mathfrak{G}_z\times\mathfrak{G}_w\right)$},\\[1mm]
            \mathcal{Z} \cap_\mathbf{R} \mathcal{V} &= 
            \scalebox{0.78}{$\left(
            \begin{bmatrix}
                \mathbf{G}_z & \mathbf{0}
            \end{bmatrix}\hspace{-1mm},
            \mathbf{c}_z, 
            \begin{bmatrix}
                \mathbf{A}_z & \mathbf{0}\\
                \mathbf{0} & \mathbf{A}_v\\
                \mathbf{R}\mathbf{G}_z & -\mathbf{G}_v
            \end{bmatrix}\hspace{-1mm},
            \begin{bmatrix}
                \mathbf{b}_z\\
                \mathbf{b}_v\\
                \mathbf{c}_v - \mathbf{R}\mathbf{c}_z
            \end{bmatrix}\hspace{-1mm},
            \mathfrak{G}_z\times\mathfrak{G}_v\right)$}.
        \end{aligned}
        \label{Eq:CCGOperations}
    \end{equation}
\end{proposition}
\setlength{\arraycolsep}{2pt}
\vspace{3mm}

In particular, CCGs generalize many commonly used set classes, such as 
intervals, ellipsoids, zonotopes, Constrained Zonotopes (CZs), polytopes, 
convex cones, ellipsotopes, and AH-polytopes. For additional details on CCGs, 
the reader is referred to \cite{silvestre2021constrained, silvestre2022accurate, silvestre2022set, silvestre2023exact}.

\section{PROBLEM FORMULATION} \label{Sec:Problem}

We consider a rigid-body ego vehicle whose occupied set in the environment is 
given by
\begin{equation}
    \mathcal{E}(\mathbf{p}, \bm{\psi}) = 
    \mathbf{p} + \mathbf{R}(\bm{\psi})\bar{\mathcal{E}},
\end{equation}
where $\mathbf{p} \in \mathbb{R}^p$ is the ego position (a fixed point in the 
body), $\bm{\psi} \in \mathbb{R}^\psi$ is an attitude representation, and 
$\mathbf{R}(\bm{\psi}) \in \mathrm{SO}(p)$ is the associated rotation matrix. 
The set $\bar{\mathcal{E}} \subset \mathbb{R}^p$ is a compact CCG describing 
the ego shape relative to $\mathbf{p}$ for $\mathbf{R}(\bm{\psi}) = \mathbf{I}$.

We will initially consider an ego with first-order dynamics:
\begin{equation}
    (\Dot{\mathbf{p}}, \Dot{\bm{\psi}}) = \mathbf{f}(\mathbf{p}, \bm{\psi})
    + \mathbf{G}(\mathbf{p}, \bm{\psi})\bm{\nu},
    \label{Eq:EgoFOD}
\end{equation}
where $\bm{\nu} \in \mathbb{R}^\nu$ is the input, the functions $\mathbf{f}$ 
and $\mathbf{G}$ are locally Lipschitz, and the first $p$ rows of 
$\mathbf{G}(\mathbf{p}, \bm{\psi})$ have full row rank.

\noindent Additionally, we will consider the case of second-order 
strict-feedback dynamics, where $\bm{\nu}$ cannot be directly manipulated and 
the ego dynamics are given by \eqref{Eq:EgoFOD} along with
\begin{equation}
    \Dot{\bm{\nu}} = \mathbf{f}_1(\mathbf{p}, \bm{\psi}, \bm{\nu})
    + \mathbf{G}_1(\mathbf{p}, \bm{\psi}, \bm{\nu})\bm{\tau},
    \label{Eq:EgoSOD}
\end{equation}
where $\bm{\tau} \in \mathbb{R}^\tau$ is the control input, the functions 
$\mathbf{f}_1$ and $\mathbf{G}_1$ are locally Lipschitz, and 
$\mathbf{G}_1(\mathbf{p}, \bm{\psi}, \bm{\nu})$ has full row rank.

The ego vehicle will navigate in an environment containing an obstacle whose 
occupied region is given by
\begin{equation}
    \mathcal{O}(\mathbf{r}) = \mathbf{r} + \bar{\mathcal{O}},
\end{equation}
where $\mathbf{r} \in \mathbb{R}^p$ is the position of the obstacle, and 
$\bar{\mathcal{O}} \subset \mathbb{R}^p$ is a compact CCG with nonempty 
interior describing the shape of the obstacle relative to $\mathbf{r}$. The 
dynamics of the obstacle are assumed to be unknown and nonlinear. 

At each sampling instant $t_k = kT_s$, with $k \in \mathbb{Z}_{\geq0}$ and 
sampling period $T_s \in \mathbb{R}_{>0}$,
a measurement $\mathbf{y}_k \in \mathbb{R}^p$ of the obstacle position is 
obtained as 
\begin{equation}
    \mathbf{y}_k = \mathbf{r}_k + \mathbf{w}_k,
    \label{Eq:yk}
\end{equation}
where $\mathbf{r}_k \in \mathbb{R}^p$ is the true obstacle position at time 
$t_k$, and $\mathbf{w}_k \in \mathbb{R}^p$ is Gaussian noise with $\mathbf{w}_k 
\overset{\text{\tiny i.d.}}{\sim} \mathcal N(\mathbf{0}, \bm{\Sigma}_k)$.

The collision-avoidance requirement is stated directly in terms of occupied 
sets. For a single obstacle, safety requires
\begin{equation}
    \mathcal{E}(\bm{\varphi}_\mathbf{p}(t), \bm{\varphi}_{\bm{\psi}}(t)) 
    \cap \mathcal{O}(\bm{\varphi}_{\mathbf{r}}(t)) = \emptyset
\end{equation}
for all $t \in \mathbb{R}_{\geq0}$, where $\bm{\varphi}_\mathbf{p}$, 
$\bm{\varphi}_{\bm{\psi}}$, and $\bm{\varphi}_{\mathbf{r}}$ denote the ego 
position, ego attitude, and obstacle trajectories, respectively.

The objective of this paper is to construct confidence obstacle occupancy sets from noisy sampled measurements, propagate them over the inter-sample interval under explicit motion bounds, and represent the resulting obstacle flow in exact CCG form. This produces a structured geometric description that serves as the basis for subsequent barrier construction and safety-filter synthesis.

\begin{remark}
    For clarity, we are assuming direct noisy measurements of the obstacle 
    position. Nonetheless, note that the framework extends to more general 
    linear sensing models of the form 
    $\mathbf{z}_k = \mathbf{H}\mathbf{r}_k + \bar{\mathbf{w}}_k$, where 
    $\mathbf{H}$ is a known full-column-rank matrix and 
    $\bar{\mathbf{w}}_k \overset{\text{\tiny i.d.}}{\sim} \mathcal 
    N(\mathbf{0}, \bar{\bm{\Sigma}}_k)$, by defining 
    $\mathbf{y}_k = \mathbf{H}^\dagger \mathbf{z}_k$, which yields the 
    equivalent measurement model \eqref{Eq:yk}. Also, the approach proposed 
    herein can be readily extended to environments containing multiple 
    obstacles. 
\end{remark}

\section{PROPOSED SOLUTION} \label{Sec:Solution}

Since the obstacle dynamics are unknown, direct model-based prediction of the obstacle position $\mathbf{r}(t)$ is not available. 
Instead, we construct a local data-driven approximation of the obstacle position using recent noisy position measurements.
The proposed solution consists of two main steps.
First, we construct a probabilistic unsafe set-valued flow for each sampling interval. 
Second, we convert the resulting unsafe set-valued flow into a CBF suitable for safety filtering.


The first step begins with a confidence set $\hat{\mathcal{R}}_k$ for the obstacle position at the current sampling instant $t_k$, obtained from a sliding-window least-squares estimator.
Details are outlined in Section \ref{Sec:SolutionA}.
To account for obstacle motion over the interval $[t_k, t_{k+1})$, we further introduce a velocity-bound-induced displacement set. 
Specifically, assume that the obstacle velocity satisfies $\Dot{\mathbf{r}} \in \mathcal{V}$, where $\mathcal{V} \subset \mathbb{R}^p$ is a compact CCG.
By combining the sampling-time confidence set $\hat{\mathcal{R}}_k$ and the velocity-induced displacement set $\mathcal{V}(t-t_k)$ with the ego and obstacle body shapes, we obtain the estimated unsafe set for the ego position. 
For any instant $t \in [t_k, t_{k+1})$ and attitude $\bm{\psi} \in \mathbb{R}^\psi$,
define the CCG-valued function $\hat{\mathcal{Q}}_k: \mathbb{R}^\psi\times[t_k, t_{k+1}) \rightrightarrows 
\mathbb{R}^p$ as
\begin{equation}
    \hat{\mathcal{Q}}_k(\bm{\psi}, t) = \hat{\mathcal{R}}_k 
    \oplus \mathcal{V}(t-t_k) 
    \oplus \left(-\mathbf{R}(\bm{\psi})\bar{\mathcal{E}}\right)
    \oplus \bar{\mathcal{O}}.
    \label{Eq:Qhat}
\end{equation}
Accordingly, over $[t_k, t_{k+1})$, the estimated unsafe 
set-valued flow for $(\mathbf{p}, \bm{\psi})$, 
$\hat{\mathcal{U}}_k: [t_k, t_{k+1}) \rightrightarrows 
\mathbb{R}^{p}\times\mathbb{R}^\psi$, is given by
\begin{equation}
    \hat{\mathcal{U}}_k(t) = \left\{(\mathbf{p}, \bm{\psi}) \in 
    \mathbb{R}^p\times\mathbb{R}^\psi: \mathbf{p} \in 
    \hat{\mathcal{Q}}_k(\bm{\psi}, t)\right\}.
    \label{Eq:Uhat}
\end{equation}

To achieve probabilistic safety over the interval $[t_k, t_{k+1})$, we then 
convert $\hat{\mathcal{U}}_k$ into a CBF $h_k: 
\mathbb{R}^p\times\mathbb{R}^\psi\times[t_k, t_{k+1}) \rightarrow \mathbb{R}$ 
for \eqref{Eq:EgoFOD}, inducing a safe flow $\mathcal{C}_k: [t_k, t_{k+1}) 
\rightrightarrows \mathbb{R}^p\times\mathbb{R}^\psi$ with $\mathcal{C}_k(t) 
\subseteq (\mathbb{R}^p\times\mathbb{R}^\psi)\setminus\hat{\mathcal{U}}_k(t)$ 
for all $t \in [t_k, t_{k+1})$. However, since 
$\hat{\mathcal{Q}}_k(\bm{\psi}, t)$ is a CCG, this conversion is nontrivial. 
The proposed conversion approach is outlined in Section \ref{Sec:SolutionB}.

When the ego has first-order dynamics as in \eqref{Eq:EgoFOD}, we can then 
directly use $h_k$ to design a safety filter $\mathbf{k}_k$ for the interval 
$[t_k, t_{k+1})$ as in \eqref{Eq:SafetyFilter}. For second-order dynamics, we 
first use $h_k$ to design a smooth safety filter $\mathbf{k}_k$ for 
\eqref{Eq:EgoFOD} (see \cite{ong2019universal}), and we then define a CBF 
$h_{k, 1}$ for \eqref{Eq:EgoFOD}-\eqref{Eq:EgoSOD} via backstepping as
\begin{equation}
    \scalebox{0.95}{$
    h_{k, 1}(\mathbf{p}, \bm{\psi}, \bm{\nu}, t) = 
    h_k(\mathbf{p}, \bm{\psi}, t) 
    - \dfrac{1}{2\mu}\|\bm{\nu} - \mathbf{k}_k(\mathbf{p}, \bm{\psi}, t)\|^2
    $},
\end{equation}
with $\mu \in \mathbb{R}_{>0}$. The function $h_{k, 1}$ can then be used to 
design a safety filter $\mathbf{k}_{k, 1}$  for 
\eqref{Eq:EgoFOD}-\eqref{Eq:EgoSOD} as in \eqref{Eq:SafetyFilter}. For further 
details on CBF backstepping, the reader is referred to \cite{taylor2022safe}, 
\cite{matias2026safe}.

\subsection{Data-Driven Probabilistic Obstacle Position Flow}
\label{Sec:SolutionA}
We estimate the obstacle position at the current sampling instant $t_k$ using the $N$ most recent noisy position measurements $\mathbf{y}_{k-N+1},\cdots,\mathbf{y}_{k}$. 
Rather than assuming a global model for the obstacle motion, we construct a local linear-in-parameters approximation over the current sliding window. 

Let $\bm{\phi}(t)=(\phi_1(t), \cdots, \phi_q(t))\in\mathbb{R}^q$ be a chosen 
smooth basis function vector. Over the current window, define the design matrix 
$\bm{\Phi}_k = [\bm{\phi}(t_{k-N+1}) \cdots \bm{\phi}(t_k)]^\top \in 
\mathbb{R}^{N\times q}$, measurement matrix $\mathbf{Y}_k = 
[\mathbf{y}_{k-N+1} \cdots \mathbf{y}_{k}]^\top 
\in \mathbb{R}^{N\times p}$, position matrix $\mathbf{R}_k = 
[\mathbf{r}_{k-N+1} \cdots \mathbf{r}_{k}]^\top \in \mathbb{R}^{N\times p}$, 
and similarly the noise matrix $\mathbf{W}_k = 
[\mathbf{w}_{k-N+1} \cdots \mathbf{w}_k]^\top \in \mathbb{R}^{N\times p}$. 
Here, the uncertainty is taken to arise solely from additive Gaussian 
measurement noise; no deterministic model-mismatch term is introduced at this 
stage.

\begin{assumption}[Local representability at $t_k$]
Over the current window, there exists a coefficient matrix 
$\bm{\Theta}_k^\ast\in\mathbb{R}^{q\times p}$ such that the noise-free stacked 
data satisfy $\mathbf{R}_k=\bm{\Phi}_k\bm{\Theta}_k^\ast$. In particular, the 
current noise-free obstacle position at $t_k$ is represented as 
$\mathbf{r}_k = \bm{\Theta}_k^{\ast\top}\bm{\phi}(t_k)$.
\end{assumption}

Under Assumption 1, we estimate the local coefficient matrix from the measured 
data by least squares:
\begin{equation}
    \hat{\bm{\Theta}}_k = \underset{\bm{\Theta} \in \mathbb{R}^{q\times p}}{\arg\min}\,\|\mathbf{Y}_{k}-\bm{\Phi}_k\bm{\Theta}\|_2^2.
    \label{eq:ls_estimator}
\end{equation}

Assuming that $\bm{\Phi}_k$ has full column rank, the least-squares problem in \eqref{eq:ls_estimator} admits the unique closed-form solution
\begin{equation}
    \hat{\bm\Theta}_k = \bm{\Phi}_k^\dagger\mathbf{Y}_{k}.
    \label{eq:theta_closed_form}
\end{equation}
The predicted obstacle position at the current sampling time is then obtained by evaluating the fitted local model at $t_k$:
\begin{equation}
    \hat{\mathbf{r}}_k=\hat{\bm\Theta}_k^\top \bm{\phi}(t_k).
    \label{eq:yhat_tk}
\end{equation}
The estimation accuracy is quantified by the prediction error
\begin{equation}
    \mathbf{e}_k = \mathbf{r}_k-\hat{\mathbf{r}}_k.
    \label{eq:prediction_error_def}
\end{equation}
Under Assumption 1, the true data satisfy $\mathbf{R}_k=\bm{\Phi}_k\bm{\Theta}_k^\ast$, while $\mathbf{Y}_{k}=\mathbf{R}_k+\mathbf{W}_k$. Substituting this relation in \eqref{eq:theta_closed_form} and evaluating \eqref{eq:yhat_tk} at $t_k$ shows that the nominal term $\bm{\Theta}_k^{\ast\top}\bm{\phi}(t_k)$ cancels in the error, yielding
\begin{equation}
    \mathbf{e}_k =
    -\mathbf{W}_k^\top \big(\bm{\Phi}_k^\dagger\big)^\top\bm{\phi}(t_k) 
    := -\mathbf{W}_k^\top \mathbf{a}_k.
    \label{eq:error_linear_representation}
\end{equation}

Since $\mathbf{e}_k$ is a linear combination of independent Gaussian vectors, it is Gaussian with zero mean and covariance
\begin{equation}
    \bm{\Pi}_k = \mathbf{a}_k^\top \bm{\Omega}_k \mathbf{a}_k,
\end{equation}
where $\bm{\Omega}_k$ is the covariance of the stacked noise, which can be represented as $\bm{\Omega}_k 
= \mathrm{diag}(\bm{\Sigma}_{k-N+1}, \dots, \bm{\Sigma}_k)$.

Since the current estimation error $\mathbf{e}_k$ is Gaussian with covariance $\bm{\Pi}_k$, the associated quadratic form $\mathbf{e}_k^\top \bm{\Pi}_k^{-1} \mathbf{e}_k $ follows a chi-square distribution with $p$ degrees of freedom.
Hence, for any confidence level $\alpha\in(0,1)$,
\begin{equation}
    \mathbb P\!\left(
    \mathbf{e}_k^\top \bm{\Pi}_k^{-1} \mathbf{e}_k \le \chi_p^2(1-\alpha)
    \right)=1-\alpha.
    \label{eq:chi_square_confidence}
\end{equation}

Thus, the obstacle position at the sampling instant $t_k$ lies in the confidence ellipsoid
\begin{equation}
    \scalebox{0.85}{$
    \hat{\mathcal{R}}_k = \left\{
    \mathbf{r}\in\mathbb R^p:
    (\mathbf{r}-\hat{\mathbf{r}}_k)^\top 
    \left(\chi_p^2(1-\alpha)\bm{\Pi}_k\right)^{-1}(\mathbf{r}-\hat{\mathbf{r}}_k)\le 1
    \right\},$}
    \label{eq:obstacle_confidence_ellipsoid}
\end{equation}
so that
\begin{equation}
    \mathbb P\!\left(\mathbf{r}_k\in\hat{\mathcal{R}}_k\right)=1-\alpha.
    \label{eq:obstacle_confidence_probability}
\end{equation}

To propagate this uncertainty over the inter-sample interval, we combine the confidence ellipsoid $\hat{\mathcal{R}}_k$ with the velocity-induced displacement set $\mathcal{V}(t-t_k)$. 
This combination yields a time-varying probabilistic obstacle position flow over $[t_k, t_{k+1})$, which corresponds to the unsafe set $\hat{\mathcal{Q}}_k(\bm{\psi}, t)$ and flow $\hat{\mathcal{U}}_k(t)$ from (\ref{Eq:Qhat}) and (\ref{Eq:Uhat}).


\subsection{CCG-to-CBF Conversion} \label{Sec:SolutionB}

As previously detailed, given a confidence CCG $\hat{\mathcal{R}}_{k}$ for the 
obstacle position, for any instant $t \in [t_k, t_{k+1})$ and attitude 
$\bm{\psi} \in \mathbb{R}^\psi$, the estimated unsafe set for the ego position 
is given by the CCG $\hat{\mathcal{Q}}_k(\bm{\psi}, t)$ defined in 
\eqref{Eq:Qhat}. Using the identities in \eqref{Eq:CCGOperations}, we then 
conclude that $\hat{\mathcal{Q}}_k(\bm{\psi}, t)$ is a CCG of the form
\begin{equation}
    \hat{\mathcal{Q}}_k(\bm{\psi}, t) = (\mathbf{G}_{\hat{q}, k}(\bm{\psi}, t), 
    \mathbf{c}_{\hat{q}, k}(\bm{\psi}, t), \mathbf{A}_{\hat{q}}, 
    \mathbf{b}_{\hat{q}}, \mathfrak{G}_{\hat{q}}),
    \label{Eq:QhatCCG}
\end{equation}
with parameters given by
\begin{equation}
    \begin{aligned}
        \mathbf{G}_{\hat{q}, k}(\bm{\psi}, t) &= 
        \begin{bmatrix}
            \mathbf{G}_{\hat{r}, k} & \mathbf{G}_{v}(t-t_k)
            & -\mathbf{R}(\bm{\psi})\mathbf{G}_{\bar{e}} & \mathbf{G}_{\bar{o}}
        \end{bmatrix},\\
        \mathbf{c}_{\hat{q}, k}(\bm{\psi}, t) &= 
        \mathbf{c}_{\hat{r}, k} + \mathbf{c}_{v}(t-t_k)
        - \mathbf{R}(\bm{\psi})\mathbf{c}_{\bar{e}}
        + \mathbf{c}_{\bar{o}},\\
        \mathbf{A}_{\hat{q}} &= \mathrm{diag}(\mathbf{A}_{\hat{r}}, 
        \mathbf{A}_{v}, \mathbf{A}_{\bar{e}}, \mathbf{A}_{\bar{o}}),\\
        \mathbf{b}_{\hat{q}} &= \left(\mathbf{b}_{\hat{r}}, \mathbf{b}_{v}, 
        \mathbf{b}_{\bar{e}}, \mathbf{b}_{\bar{o}}\right),\\
        \mathfrak{G}_{\hat{q}} &= \mathfrak{G}_{\hat{r}}\times
        \mathfrak{G}_{v}\times\mathfrak{G}_{\bar{e}}\times
        \mathfrak{G}_{\bar{o}}.
    \end{aligned}
\end{equation}
As all sets in \eqref{Eq:Qhat} are compact and at least one has nonempty 
interior, $\hat{\mathcal{Q}}_k(\bm{\psi}, t)$ is also compact and has nonempty 
interior. Given $\hat{\mathcal{Q}}_k$, the next step is then to translate the 
unsafe set-valued flow $\hat{\mathcal{U}}_k$ defined in \eqref{Eq:Uhat} into a 
CBF $h_k$ for \eqref{Eq:EgoFOD}. To this end, we build on the approach proposed 
in \cite{matias2026safe}, generalizing it to handle vehicle attitude, which was 
not addressed in \cite{matias2026safe}. 

Let $\bm{\xi}$ be the generator variable of 
$\hat{\mathcal{Q}}_k(\bm{\psi}, t)$. Our first step is to eliminate the 
equality constraint $\mathbf{A}_{\hat{q}}\bm{\xi} = \mathbf{b}_{\hat{q}}$ in 
\eqref{Eq:QhatCCG}. As $\hat{\mathcal{Q}}_k(\bm{\psi}, t)$ has nonempty 
interior, this equation has infinitely many solutions, which can be expressed 
parametrically as
\begin{equation}
    \bm{\xi} = \mathbf{A}_{\hat{q}}^\dagger\mathbf{b}_{\hat{q}} 
    + \mathrm{null}(\mathbf{A}_{\hat{q}})\bm{\eta},\quad 
    \bm{\eta} \in \mathbb{R}^\eta.
\end{equation}
Next, observe that, by construction, the set $\mathfrak{G}_{\hat{q}}$ 
can be written as $\mathfrak{G}_{\hat{q}} = 
\mathcal{G}_1\times\dots\times\mathcal{G}_G$, where, for each 
$i \in \mathcal{I} = \{1, \dots, G\}$, $\mathcal{G}_i$ is the zero-sublevel set 
of a convex function $g_i: \mathbb{R}^{\xi_i} \rightarrow \mathbb{R}$. This 
allows $\mathfrak{G}_{\hat{q}}$ to be written compactly as
\begin{equation}
    \mathfrak{G}_{\hat{q}} = \left\{\bm{\xi} \in \mathbb{R}^{\xi}: 
    \max_{i \in \mathcal{I}}\,g_i(\mathbf{S}_i\bm{\xi}) \leq 0\right\},
\end{equation}
where $\mathbf{S}_i \in \mathbb{R}^{\xi_i\times\xi}$ extracts the $i$th 
generator vector $\bm{\xi}_i$ from the full vector 
$\bm{\xi} = (\bm{\xi}_1, \dots, \bm{\xi}_G)$. However, as the maximum operator 
is not differentiable, it is unsuitable for CBF design. To address this, we 
employ a smooth underapproximation of the maximum, adopting the LogSumExp 
approach from \cite{molnar2023composing}.

Combining the above two steps, we introduce a new CCG-valued function 
$\tilde{\mathcal{Q}}_k$ defined as
\begin{equation}
    \tilde{\mathcal{Q}}_k(\bm{\psi}, t) = 
    (\mathbf{G}_{\tilde{q}, k}(\bm{\psi}, t), 
    \mathbf{c}_{\tilde{q}, k}(\bm{\psi}, t), [\,\,], [\,\,], 
    \mathfrak{G}_{\tilde{q}}),
\end{equation}
whose parameters are given by
\begin{equation}
    \begin{aligned}
        \mathbf{G}_{\tilde{q}, k}(\bm{\psi}, t) &= 
        \mathbf{G}_{\hat{q}, k}(\bm{\psi}, t)
        \mathrm{null}(\mathbf{A}_{\hat{q}}),\\
        \mathbf{c}_{\tilde{q}, k}(\bm{\psi}, t) &=  
        \mathbf{c}_{\hat{q}, k}(\bm{\psi}, t) 
        + \mathbf{G}_{\hat{q}, k}(\bm{\psi}, t)
        \mathbf{A}_{\hat{q}}^\dagger\mathbf{b}_{\hat{q}}.
    \end{aligned}
\end{equation}
Moreover, the generator set $\mathfrak{G}_{\tilde{q}}$ is the zero-sublevel set 
of a function $f: \mathbb{R}^{\eta} \rightarrow \mathbb{R}$ defined as
\begin{equation}
    \begin{aligned}
        f(\bm{\eta}) &= \frac{1}{\gamma}\ln\left(\sum_{i \in \mathcal{I}}
        \exp(\gamma f_i(\bm{\eta}))\right)- \frac{\ln(G+1)}{\gamma},
    \end{aligned}
\end{equation}
where each function $f_i: \mathbb{R}^{\eta} \rightarrow \mathbb{R}$ is defined 
as
\begin{equation}
    f_i(\bm{\eta}) = 
    g_i\left(\mathbf{S}_i\mathbf{A}_{\hat{q}}^\dagger\mathbf{b}_{\hat{q}}
    + \mathbf{S}_i\mathrm{null}(\mathbf{A}_{\hat{q}})\bm{\eta}\right),
\end{equation}
and \hfill $\gamma \in \mathbb{R}_{>0}$ \hfill is \hfill a \hfill smoothing 
\hfill parameter. \hfill This \hfill construction 

\newpage

\noindent eliminates the equality constraint and replaces the maximum operator 
with a smooth approximation, yielding a representation suitable for CBF design. 
Furthermore, $\tilde{\mathcal{Q}}_k$ induces an approximation 
$\tilde{\mathcal{U}}_k$ of $\hat{\mathcal{U}}_k$, which satisfies 
$\hat{\mathcal{U}}_k(t) \subset \mathrm{int}(\tilde{\mathcal{U}}_k(t))$ 
for all $t \in [t_k, t_{k+1})$ and converges to $\hat{\mathcal{U}}_k$ as 
$\gamma\rightarrow\infty$ \cite{molnar2023composing}.

Our goal is now to represent $\tilde{\mathcal{U}}_k(t)$ using a CBF candidate 
$h_k$. To this end, observe that a state $(\mathbf{p}, \bm{\psi})$ belongs to 
$\tilde{\mathcal{U}}_k(t)$ if and only if there exists 
$\bm{\eta} \in \mathbb{R}^\eta$ such that $f(\bm{\eta}) \leq 0$ and 
$\mathbf{G}_{\tilde{q}, k}(\bm{\psi}, t)\bm{\eta} 
+ \mathbf{c}_{\tilde{q}, k}(\bm{\psi}, t) = \mathbf{p}$. Thus, 
$\tilde{\mathcal{U}}_k(t)$ can be written as
\begin{equation}
     \tilde{\mathcal{U}}_k(t) = \{(\mathbf{p}, \bm{\psi}) \in 
    \mathbb{R}^p\times\mathbb{R}^\psi: h_k(\mathbf{p}, \bm{\psi}, t) \leq 0\},
     \label{Eq:CBFCCGFlow}
\end{equation}
where $h_k$ is defined via the following optimization problem:
\begin{equation} 
    \begin{aligned}
        h_k(\mathbf{p}, \bm{\psi}, t) = 
        \min_{\bm{\eta} \in \mathbb{R}^\eta}\,\, &f(\bm{\eta})\\
        \text{subject to}\,\, 
        &\mathbf{G}_{\tilde{q}, k}(\bm{\psi}, t)\bm{\eta} 
        + \mathbf{c}_{\tilde{q}, k}(\bm{\psi}, t) = \mathbf{p}.
    \end{aligned}
    \label{Eq:CBFCCG}
\end{equation}
The following 
result provides a mild sufficient condition for $h_k$ to be a valid CBF for 
\eqref{Eq:EgoFOD}.

\begin{theorem}
    If $f$ is twice continuously differentiable and strictly convex, then $h_k$ 
    is a valid CBF for \eqref{Eq:EgoFOD}.
\end{theorem}

\begin{proof}
    The proof is similar to that of Theorem 3 from \cite{matias2026safe}, with 
    the added technicality of the attitude dependence.

    To show that $h_k$ is a CBF for \eqref{Eq:EgoFOD}, we must verify that: (i) 
    $h_k$ is continuously differentiable, (ii) 
    $\frac{\partial}{\partial(\mathbf{p}, \bm{\psi})}
    h_k(\mathbf{p}, \bm{\psi}, t) \neq \mathbf{0}$ for 
    $h_k(\mathbf{p}, \bm{\psi}, t) = 0$, and (iii) there is an extended 
    class-$\mathcal{K}_\infty$ $\alpha$ such that 
    $\sup_{\bm{\nu} \in \mathbb{R}^\nu} 
    \Dot{h}_k(\mathbf{p}, \bm{\psi}, t, \bm{\nu}) 
    > -\alpha(h_k(\mathbf{p}, \bm{\psi}, t))$ when 
    $h_k(\mathbf{p}, \bm{\psi}, t) \geq 0$ (at least).
    
    As the optimization problem in \eqref{Eq:CBFCCG} is convex, its 
    solutions coincide with those of the corresponding KKT system:
    \begin{equation}
        \scalebox{0.91}{$
        \bm{\gamma}_k(\mathbf{p}, \bm{\psi}, t, \bm{\eta}, \bm{\lambda}) :=
        \begin{bmatrix}
            \left(\dfrac{\partial f(\bm{\eta})}{\partial\bm{\eta}}\right)^\top
            + \mathbf{G}_{\tilde{q}, k}(\bm{\psi}, t)^\top\bm{\lambda}\\[3mm]
            \mathbf{G}_{\tilde{q}, k}(\bm{\psi}, t)\bm{\eta} + 
            \mathbf{c}_{\tilde{q}, k}(\bm{\psi}, t) - \mathbf{p}
        \end{bmatrix} = \mathbf{0},$}
        \label{Eq:KKT}
    \end{equation}
    with $\bm{\lambda} \in \mathbb{R}^{p}$. The derivative of $\bm{\gamma}_k$ 
    with respect to $(\bm{\eta}, \bm{\lambda})$ is
    \begin{equation}
        \frac{\partial \bm{\gamma}_k(\mathbf{p}, \bm{\psi}, t, \bm{\eta}, 
        \bm{\lambda})}{\partial (\bm{\eta}, \bm{\lambda})} =
        \begin{bmatrix}
            \dfrac{\partial^2 f(\bm{\eta})}{\partial\bm{\eta}^2}
            & \mathbf{G}_{\tilde{q}, k}(\bm{\psi}, t)^\top\\[3mm]
            \mathbf{G}_{\tilde{q}, k}(\bm{\psi}, t) & \mathbf{0}
        \end{bmatrix}.
    \end{equation}
    As $f$ is strictly convex and $\mathbf{G}_{\tilde{q}, k}(\bm{\psi}, t)$ has 
    full row rank (as $\mathrm{int}(\tilde{\mathcal{U}}_k(t)) \neq \emptyset$), 
    the above matrix is invertible, by the Schur complement. Hence, by the 
    Implicit Function Theorem, the KKT system defines 
    $(\bm{\eta}, \bm{\lambda}) = \bm{\ell}_k(\mathbf{p}, \bm{\psi}, t)$, where 
    the implicit function $\bm{\ell}_k$ is continuously differentiable with
    \begin{equation}
        \scalebox{1.12}{$
        \frac{\partial \bm{\ell}_k(\mathbf{p}, \bm{\psi}, t)}
        {\partial(\mathbf{p}, \bm{\psi}, t)} 
        = - \left[\left(\frac{\partial \bm{\gamma}_k(\cdot)}
        {\partial (\bm{\eta}, \bm{\lambda})}\right)^{-1} 
        \frac{\partial \bm{\gamma}_k(\cdot)}
        {\partial (\mathbf{p}, \bm{\psi}, t)}\right]_{
        \bm{\ell}_k(\mathbf{p}, \bm{\psi}, t)}.$}
    \end{equation}
    Thus, $h_k$ can be written as $h_k(\mathbf{p}, \bm{\psi}, t) = 
    f(\mathbf{E}_{\bm{\eta}}\bm{\ell}_k(\mathbf{p}, \bm{\psi}, t))$, where 
    $\mathbf{E}_{\bm{\eta}}$ extracts $\bm{\eta}$ from 
    $(\bm{\eta}, \bm{\lambda})$, from which it is clear that $h_k$ is 
    continuously differentiable, hence proving (i).

    Since $h_k$ is convex in $\mathbf{p}$, we have 
    $\frac{\partial}{\partial\mathbf{p}} h_k(\mathbf{p}, \bm{\psi}, t) \neq 
    \mathbf{0}$ when $h_k(\mathbf{p}, \bm{\psi}, t) \geq 0$. Thus, (ii) is 
    satisfied, and (iii) holds for any extended class-$\mathcal{K}_\infty$ 
    $\alpha$ since $\Dot{h}_k$ can be made arbitrarily large when 
    $\frac{\partial}{\partial\mathbf{p}} h_k(\mathbf{p}, \bm{\psi}, t) \neq 
    \mathbf{0}$, by the rank of $\mathbf{G}(\mathbf{p}, \bm{\psi})$.
\end{proof}

The previous result gives a mild condition for $h_k$ to be a CBF for 
\eqref{Eq:EgoFOD}. Particularly, if each $g_i$ is twice continuously 
differentiable and strictly convex, then $f$ inherits these properties and is a 
valid CBF for \eqref{Eq:EgoFOD}. For details on constructing CCGs satisfying such conditions, we refer the reader to \cite{matias2026safe}.

\section{SIMULATION RESULTS} \label{Sec:Results}

We consider a planar navigation scenario involving a fully actuated 3-DOF ego 
vehicle tracking a smooth reference trajectory 
$\mathbf{q}: \mathbb{R}_{\geq0} \rightarrow \mathbb{R}^2$ generated via 
waypoint interpolation, while avoiding a nonlinear obstacle vehicle 
performing a circular trajectory. The ego has polytopic geometry,
and the obstacle has ellipsoidal geometry that depends on its attitude. As 
obstacle attitude is not measured, avoidance is based on a circular body set 
$\bar{\mathcal{O}}$ obtained by revolving the ellipsoid over all possible 
attitudes. 
The obstacle position is measured as
\begin{equation}
    \mathbf{y}_k = \mathbf{r}_k + \mathbf{w}_k,\quad 
    \mathbf{w}_k \overset{\text{\tiny i.i.d.}}{\sim} 
    \mathcal N(\mathbf{0}, \sigma^2\mathbf{I}).
\end{equation}

The simulations are conducted with $T_s = \SI{0.1}{s}$, $\sigma = 0.5$. The 
data-driven approach from Section \ref{Sec:SolutionA} is implemented with a 
third-order polynomial approximation, window length $N = 100$, and a confidence 
level of $\SI{95}{\%}$, i.e., $\alpha = 0.05$.

\begin{figure}[t]
    \centering
    \includegraphics[width=0.9\linewidth]{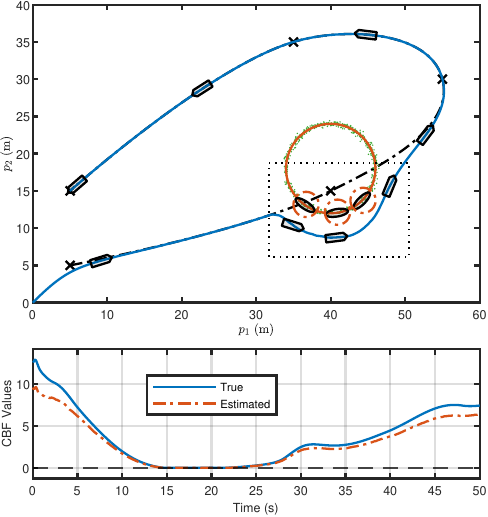}
    \caption{Safe navigation of a polytopic ego vehicle with first-order 
    dynamics around an obstacle vehicle moving along a circular trajectory. The 
    top panel shows the trajectories of the ego vehicle (blue) and the 
    obstacle (orange), where the orange dash-dotted circles illustrate the 
    estimated obstacle sets. The three snapshots inside the dotted rectangle 
    depict consecutive instants. The bottom panel depicts the estimated and 
    true CBF values over time.}
    \label{Fig:Example1}
\end{figure}

\begin{example}
    We begin with an ego vehicle with first-order dynamics of the form
    \begin{equation}
        (\Dot{\mathbf{p}}, \Dot{\psi}) 
        = \mathbf{G}(\mathbf{p}, \psi)\bm{\nu},\quad
        \mathbf{G}(\mathbf{p}, \psi) = 
        \begin{bmatrix}
            \mathbf{R}(\psi) & \mathbf{0}_{2\times1}\\
            \mathbf{0}_{1\times2} & 1
        \end{bmatrix},
        \label{Eq:ExEgoFOD}
    \end{equation}
    where $\psi \in \mathbb{R}$ is the heading angle and 
    $\mathbf{R}(\psi) \in \mathrm{SO}(2)$ the respective rotation matrix, and 
    $\bm{\nu} \in \mathbb{R}^3$ is the input, composed of the 
    body-frame surge, sway, and yaw velocities. Reference tracking is achieved 
    with a nominal controller $\mathbf{k}_\mathrm{d}$ defined as
    \begin{equation}
        \mathbf{k}_\mathrm{d}(\mathbf{p}, \psi, t) 
        = \mathbf{G}(\mathbf{p}, \psi)^\top
        \begin{bmatrix}
            \Dot{\mathbf{q}}(t) + k_\mathbf{p}(\mathbf{q}(t)-\mathbf{p})\\
            k_\psi(\theta(t)-\psi)
        \end{bmatrix},
    \end{equation}
    with $\theta(t) = \mathrm{atan2}
    (\Dot{\mathbf{q}}(t) + k_\mathbf{p}(\mathbf{q}(t)-\mathbf{p}))$ and 
    $k_\mathbf{p}, k_\psi \in \mathbb{R}_{>0}$.

    Fig. \ref{Fig:Example1} illustrates this scenario, showing that the ego 
    vehicle successfully \hfill avoids \hfill the \hfill obstacle \hfill while 
    \hfill continuing \hfill to \hfill follow
    
    \newpage
    
    \noindent the reference path. The top panel shows the resulting avoidance 
    maneuver around the estimated unsafe region, and the bottom panel indicates 
    that both the true and estimated CBF values approach zero near the closest 
    interaction but remain nonnegative, confirming that safety is maintained.
\end{example}

\begin{example}
    In this second example, the ego has second-order dynamics, extending the 
    first-order model \eqref{Eq:ExEgoFOD} with
    \begin{equation}
        \Dot{\bm{\nu}} = \mathbf{M}^{-1}(\bm{\tau}-\mathbf{D}\bm{\nu}),
    \end{equation}
    where $\bm{\tau} \in \mathbb{R}^3$ is the input, and 
    $\mathbf{M}, \mathbf{D} \in \mathbb{R}^{3\times3}$ represent the dynamics 
    of an Unmanned Surface Vehicle (USV). Tracking is achieved with a nominal 
    controller $\mathbf{k}_{\mathrm{d}, 1}$ defined as
    \begin{equation}
        \mathbf{k}_{\mathrm{d}, 1}(\mathbf{p}, \psi, \bm{\nu}, t) 
        = \mathbf{K}_{\bm{\nu}}(
        \mathbf{k}_\mathrm{d}(\mathbf{p}, \psi, t)-\bm{\nu}),
    \end{equation}
    with a positive-definite gain matrix 
    $\mathbf{K}_{\bm{\nu}} \in \mathbb{R}^{3\times3}$. 
    
    Fig. \ref{Fig:Example2} presents the resulting behavior of the ego vehicle. As in the previous example, the vehicle maintains safe separation from the moving obstacle while continuing to track the reference trajectory. The top panel shows the resulting motion relative to the estimated unsafe region, illustrating that the vehicle adopts a more conservative avoidance strategy. Unlike the first-order case, the additional inertia associated with the USV-like dynamics prevents the ego vehicle from executing a sharp maneuver around the obstacle. Instead, it slows down and temporarily drifts backward before resuming convergence to the reference path. The bottom panel shows that both the true and estimated CBF values approach zero at the point of closest interaction but remain nonnegative throughout the maneuver, confirming that the safety constraint is satisfied at all times.
\end{example}

\begin{figure}[t]
    \centering
    \includegraphics[width=0.9\linewidth]{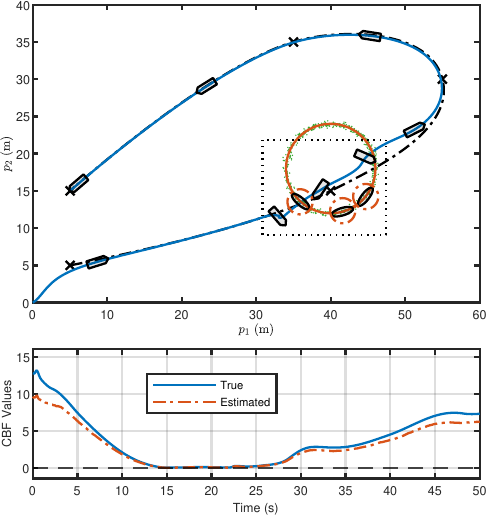}
    \caption{Navigation of a polytopic ego vehicle with second-order dynamics 
    around an obstacle vehicle moving along a circular trajectory. The top 
    panel shows the trajectories of the ego vehicle (blue) and the obstacle 
    (orange), where the orange dash-dotted circles illustrate the estimated 
    obstacle sets. The three snapshots inside the dotted rectangle depict 
    consecutive instants. The bottom panel depicts the estimated and true CBF 
    values over time.}
    \label{Fig:Example2}
\end{figure}

\section{CONCLUSION} \label{Sec:Conclusion}
This letter proposed a CBF-based probabilistic safe-navigation framework for moving obstacles with unknown nonlinear dynamics. 
From noisy sampled measurements, it constructs and propagates a probabilistic unsafe set over the inter-sample interval, represents it exactly in CCG form, and converts it into a CBF for safety-filter synthesis. 
Simulations verify safe navigation for both first- and second-order ego dynamics under measurement uncertainty and unknown obstacle motion.

\addtolength{\textheight}{-3cm}   


\bibliographystyle{ieeetr}
\bibliography{Refs}

\end{document}